\documentclass[a4paper]{article}

\usepackage[dvipdfmx]{hyperref}
\usepackage[dvipdfmx]{graphicx,xcolor}
\usepackage{amsmath,amsfonts,amssymb,amsthm}

\title{Long period sequences generated by the logistic map over finite fields with control parameter four}
\author{Kazuyoshi Tsuchiya \thanks{The author is with the Koden Electronics Co., Ltd., Ota-ku, 146-0095 Japan} \and
Yasuyuki Nogami \thanks{The author is with the Graduate School of Natural Science and Technology, Okayama University,
Okayama-shi, 700-8530 Japan}
\date{}}

\newtheorem{theorem}{Theorem}
\newtheorem{proposition}[theorem]{Proposition}
\newtheorem{lemma}[theorem]{Lemma}
\newtheorem{corollary}[theorem]{Corollary}

\newtheorem{example}[theorem]{Example}
\newtheorem{remark}[theorem]{Remark}

\begin{document}
\maketitle
\begin{abstract}
Pseudorandom number generators have been widely used in Monte Carlo methods, communication systems,
cryptography and so on.
For cryptographic applications, pseudorandom number generators are required to generate sequences
which have good statistical properties, long period and unpredictability.
A Dickson generator is a nonlinear congruential generator whose recurrence function is the Dickson polynomial.
Aly and Winterhof obtained a lower bound on the linear complexity profile of a Dickson generator.
Moreover Vasiga and Shallit studied the state diagram given by the Dickson polynomial of degree two.
However, they do not specify sets of initial values which generate a long period sequence.
In this paper, we show conditions for parameters and initial values to generate long period sequences, and
asymptotic properties for periods by numerical experiments.
We specify sets of initial values which generate a long period sequence.
For suitable parameters, every element of this set occurs exactly once
as a component of generating sequence in one period.
In order to obtain sets of initial values,
we consider a logistic generator proposed by Miyazaki, Araki, Uehara and Nogami,
which is obtained from a Dickson generator of degree two with a linear transformation.
Moreover, we remark on the linear complexity profile of the logistic generator.
The sets of initial values are described by values of the Legendre symbol.
The main idea is to introduce a structure of a hyperbola to the sets of initial values.
Our results ensure that generating sequences of Dickson generator of degree two have long period.
As a consequence, the Dickson generator of degree two has some good properties for cryptographic applications.
\end{abstract}

\section{Introduction}
\label{section:Introduction}

Pseudorandom number generators have been widely used in Monte Carlo methods, communication systems,
cryptography and so on.
For cryptographic applications, pseudorandom number generators are required to generate sequences
which have good statistical properties, long period and unpredictability.
Therefore nonlinearity of a state transition function is important.

Linear complexity profile of a sequence is a measure of nonlinearity.
Gutierrez, Shparlinski and Winterhof \cite{Gutierrez-Shparlinski-Winterhof} gave a lower bound
on the linear complexity profile of a general nonlinear congruential generator.
For some special recurrence functions, much better results were shown, namely,
the inverse functions \cite{Gutierrez-Shparlinski-Winterhof},
the power functions \cite{Shparlinski,Griffin-Shparlinski},
the Dickson polynomials \cite{Aly-Winterhof} and
the R\'edei functions \cite{Meidl-Winterhof}
(See \cite{Winterhof} and \cite[10.4.4.2]{Mullen-Panario}).

A Dickson generator is a nonlinear congruential generator whose recurrence function is the Dickson polynomial.
A lower bound on the linear complexity profile of a Dickson generator
was given by Aly and Winterhof \cite{Aly-Winterhof}.
Moreover the state diagram was studied by Vasiga and Shallit \cite{Vasiga-Shallit} in the case of degree $2$.
Vasiga and Shallit obtained a period of a sequence for any initial value,
and showed the structure of the state diagram.
However, they do not specify sets of initial values which generate a long period sequence.

In this paper, we show conditions for parameters and initial values to generate long period sequences, and
asymptotic properties for periods by numerical experiments.
We specify sets of initial values which generate a long period sequence.
For suitable parameters, every element of this set occurs exactly once
as a component of generating sequence in one period.
In order to obtain sets of initial values,
we consider a logistic generator with control parameter four,
which can be obtained from a Dickson generator of degree $2$ with a linear transformation.
A logistic generator was proposed by Miyazaki, Araki, Uehara and Nogami.
They showed properties of the generator in
\cite{Miyazaki-Araki-Uehara-Nogami:JSIAM2013,Miyazaki-Araki-Uehara-Nogami:SCIS2013,Miyazaki-Araki-Uehara-Nogami:ISITA2014,Miyazaki-Araki-Uehara-Nogami:SCIS2014,Miyazaki-Araki-Uehara-Nogami:SCIS2015}.
The sets of initial values are described by values of the Legendre symbol.
The main idea is to introduce a structure of a hyperbola to the sets of initial values.
Although we use a structure of a hyperbola to give conditions for initial values and parameters to have a long period,
we can use this structure to improve a lower bound on the linear complexity profile
of a Dickson generator of degree $2$.

This paper is organized as follows:
In Sect.\ref{section:The Dickson generator of degree $2$},
we introduce the definition of the Dickson generator and known results for the Dickson generator of degree $2$.
In Sect.\ref{section:The logistic generator},
we introduce the definition of the logistic map over finite fields and the logistic generator.
In Sect.\ref{section:Periods of logistic generator sequences},
we consider periods of logistic generator sequences with control parameter four.
In particular, we show conditions for parameters to be maximal on certain sets of initial values, and
asymptotic properties for periods by numerical experiments.
In Sect.\ref{section:Remarks on the linear complexity profile of the logistic generator},
we remark on the linear complexity profile of the logistic generator.
In Sect.\ref{section:A possibility of generalization},
we consider a possibility of generalization.
Finally, we describe some conclusions in Sect.\ref{section:Conclusion}.

We give some notations.
For a prime number $p$, $\mathbb{F}_p$ denotes the finite field with $p$ elements.
For a prime number $p$ and an integer $a$, $(a / p)$ denotes the Legendre symbol.
$D_0$ and $D_1$ denote the set of non-zero quadratic residues modulo $p$
and of quadratic non-residues modulo $p$, respectively.
For an integer $b$, put
\begin{equation*}
D_i - b = \{ c \in \mathbb{F}_p \mid c + b \in D_i \}, \quad i = 0, 1.
\end{equation*}
For a finite set $A$, $\# A$ denotes the number of elements in $A$.
For a finite field $\mathbb{F}_p$ and the quadratic extension $\mathbb{F}_{p^2}$ of $\mathbb{F}_p$,
$\mathrm{N}_{\mathbb{F}_{p^2} / \mathbb{F}_p } : \mathbb{F}_{p^2} \rightarrow \mathbb{F}_p$ denotes the norm map,
namely, $\mathrm{N}_{\mathbb{F}_{p^2} / \mathbb{F}_p } ( \alpha ) = \alpha {\alpha}^p$
for $\alpha \in \mathbb{F}_{p^2}$.
For a finite group $G$ and an element $g \in G$, $\mathrm{ord}_G \, g$ denotes the order of $g$ in $G$.
In particular,
if $G$ is the multiplicative group $( \mathbb{Z} / N \mathbb{Z} )^{\times}$
of the quotient ring $\mathbb{Z} / N \mathbb{Z}$ of integers modulo $N$,
then we write $\mathrm{ord}_N \, g$.
Let $S = ( s_n )_{n \geq 0}$ be an eventually periodic sequence.
Then there exists the least positive integer $r = r(S)$ such that $s_r \in \{ s_0, \dots ,s_{r - 1} \}$.
Let $t = t(S)$ be the least non-negative integer such that $s_r = s_t$.
If $t > 0$, we call $( s_0, \dots ,s_{t - 1} )$ the tail of $S$.
We call $( s_t, \dots ,s_{r - 1} )$ the cycle of $S$.
Let $c(S) = r - t$. We call $c(S)$ the period of $S$.

A preliminary version \cite{Tsuchiya-Nogami} of this paper was presented at
the seventh International Workshop on Signal Design and its Applications in Communications (IWSDA 2015).

\section{The Dickson generator of degree $2$}
\label{section:The Dickson generator of degree $2$}

In this section, we introduce known results for the Dickson generator of degree $2$
by Vasiga and Shallit \cite{Vasiga-Shallit}.

\subsection{The Dickson generator}

We recall the definition of the Dickson polynomials.
For details, see \cite[9.6]{Mullen-Panario}.
Let $p$ be a prime number.
For $a \in \mathbb{F}_p$,
the family of Dickson polynomials $D_e (X, a) \in \mathbb{F}_p [X]$ is defined as the recurrence relation
$D_e (X, a) = X D_{e - 1} (X, a) - a D_{e - 2} (X, a), \, e \geq 2$
with initial values $D_0 (X, a) = 2$ and $D_1 (X, a) = X$.
Then the degree of $D_e (X, a)$ is $e$ for any $e \geq 0$.
Note that $F_u (X) = X^2 - u X + 1 \in \mathbb{F}_p [X]$
is the characteristic polynomial of the family $D_e (u, 1), \, e \geq 0$ for $u \in \mathbb{F}_p$.

For $s_0 \in \mathbb{F}_p$, the sequence $S = ( s_n )_{n \geq 0}$ defined as the recurrence relation
\begin{equation}
s_{n + 1} = D_e ( s_n , 1), \qquad n \geq 0
\label{equation:Dickson generator}
\end{equation}
is called a Dickson generator sequence of degree $e$.

\subsection{The Dickson generator of degree $2$}
\label{subsection:The Dickson generator of degree $2$}

By the definition, $D_2 (X, 1) = X^2 - 2$.
Vasiga and Shallit \cite{Vasiga-Shallit} studied the state diagram of the Dickson generator of degree $2$.
They showed a period of a sequence for any initial value.
For an odd integer $m$, define $\mathrm{ord}_m' \, 2$ to be the least positive integer $k$
such that $2^k \equiv \pm 1 \mathrm{~mod~} m$.

\begin{theorem}[\cite{Vasiga-Shallit} Theorem 12]
Let $S = ( s_n )_{n \geq 0}$ be a Dickson generator sequence of degree $2$.
Let $\alpha$ and $\beta$ be the roots of $F_u (X) = X^2 - s_0 X + 1$.
Let $\mathrm{ord}_{G} \, \alpha = 2^e \cdot m$, where $G = \mathbb{F}_{p^2}^{\times}$ and $m$ is odd.
Then the length of the tail of $S$ is $e$, and the period of $S$ is $\mathrm{ord}_m' \, 2$.
\label{theorem:Vasiga-Shallit Theorem 12}
\end{theorem}

\begin{theorem}[\cite{Vasiga-Shallit} Corollary 15]
Let $p > 2$ be a prime number.
Suppose that $p - 1 = 2^f \cdot m, p + 1 = 2^{f'} \cdot m'$, where $m$ and $m'$ are odd integers.
For any divisor $d \neq 1$ of $m$, the state diagram given by $D_2 (X, 1) = X^2 - 2$ consists of
$\varphi(d) / (2 \mathrm{ord}_d' \, 2)$ cycles of period $\mathrm{ord}_d' \, 2$,
where $\varphi$ is Euler's totient function.
A complete binary tree of height $f - 1$ is attached to each element in these cycles.
Similarly, for any divisor $d' \neq 1$ of $m'$, the state diagram given by $D_2 (X, 1) = X^2 - 2$ consists of
$\varphi(d') / (2 \mathrm{ord}_{d'}' \, 2)$ cycles of period $\mathrm{ord}_{d'}' \, 2$.
A complete binary tree of height $f' - 1$ is attached to each element in these cycles.
Finally, the element $0$ is the root of a complete binary tree of height $f - 2$ (resp. $f' - 2$)
if $p \equiv 1 \mathrm{~mod~} 4$ (resp. $p \equiv 3 \mathrm{~mod~} 4$).
The state diagram consists of the directed edges $(0, - 2), (- 2, 2), (2, 2)$.
\label{theorem:Vasiga-Shallit Corollary 15}
\end{theorem}

Although they showed the structure of the state diagram given by $D_2 (X, 1) = X^2 - 2$,
they do not specify sets of initial values which generate a long period sequence.
In order to obtain these sets, we introduce the logistic generator proposed by Miyazaki, Araki, Uehara and Nogami.

\section{The logistic generator}
\label{section:The logistic generator}

In this section, we introduce the logistic map over finite fields and the logistic generator.

Let $p$ be a prime number and $\mu_p \in \mathbb{F}_p - \{ 0 \}$.
The logistic map $\mathrm{LM}_{\mathbb{F}_p [\mu_p]} : \mathbb{F}_p \rightarrow \mathbb{F}_p$
over $\mathbb{F}_p$ with control parameter $\mu_p$ is defined as 
$\mathrm{LM}_{\mathbb{F}_p [\mu_p]} (a) = \mu_p a (a + 1)$ for any $a \in \mathbb{F}_p$
(See \cite{Miyazaki-Araki-Uehara-Nogami:SCIS2013} or \cite{Miyazaki-Araki-Uehara-Nogami:ISITA2014}).
If $p > 3$ and $\mu_p = 4$, it is simply referred to as $\mathrm{LM}_{\mathbb{F}_p}$.

Assume that $p > 3$.
For $s_0 \in \mathbb{F}_p$, the sequence $S = ( s_n )_{n \geq 0}$ defined as the recurrence relation
\begin{equation}
s_{n + 1} = \mathrm{LM}_{\mathbb{F}_p} (s_n), \qquad n \geq 0
\label{equation:logistic generator}
\end{equation}
is called a logistic generator sequence.
Note that a logistic generator is a kind of quadratic congruential generator modulo $p$.
For any $n \geq 0$, put ${s_n}' = 4 s_n + 2$.
Then $S' = ( {s_n}' )_{n \geq 0}$ is a Dickson generator sequence of degree $2$.

\section{Periods of logistic generator sequences}
\label{section:Periods of logistic generator sequences}

In this section, we consider periods of logistic generator sequences with control parameter four.
In particular, we show conditions for parameters to be maximal on certain sets of initial values, and
asymptotic properties for periods by numerical experiments.

Throughout this section, let $p > 3$ be a prime number.

\subsection{The sets of initial values and a structure of the hyperbola}

In order to obtain conditions for sets of initial values,
we observe examples of state diagrams given by $\mathrm{LM}_{\mathbb{F}_p}$.

\begin{example}
Figure \ref{figure:state diagram of p = 17} (resp. Figure \ref{figure:state diagram of p = 23})
describes the state diagram given by $\mathrm{LM}_{\mathbb{F}_p}$ in the case of $p = 17$ (resp. $p = 23$).
Here,
an integer $a$ in a circle means that $( a / p ) = 1$,
an integer $a$ in a rectangle means that $( a / p ) = -1$ and
an integer $a$ in a triangle means that $a \equiv 0 \mathrm{~mod~} p$.

\begin{figure}[tb]
\begin{center}
\includegraphics[width=2.5in]{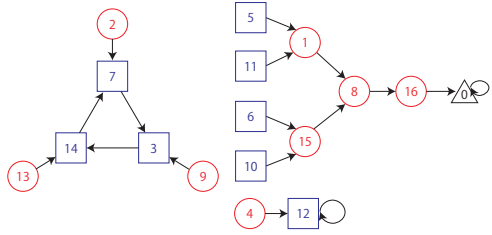}
\end{center}
\caption{The state diagram given by $\mathrm{LM}_{\mathbb{F}_p}$
with information on values of the Legendre symbol in the case of $p = 17$.}
\label{figure:state diagram of p = 17}
\end{figure}

\begin{figure}[tb]
\begin{center}
\includegraphics[width=2.5in]{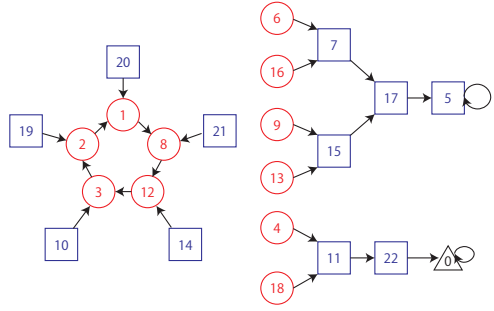}
\end{center}
\caption{The state diagram given by $\mathrm{LM}_{\mathbb{F}_p}$
with information on values of the Legendre symbol in the case of $p = 23$.}
\label{figure:state diagram of p = 23}
\end{figure}
\end{example}

It is observed that $\# \{ \mathrm{LM}_{\mathbb{F}_p}^{-1} (a) \} = 0$ if and only if
$( a / p ) \neq ( \mathrm{LM}_{\mathbb{F}_p} (a) / p )$ for $a \in \mathbb{F}_p - \{ - 1 \}$.
In the case of $p = 17$ (resp. $p = 23$), an element $a \in \mathbb{F}_p$ such that
$( a / p ) = ( \mathrm{LM}_{\mathbb{F}_p} (a) / p ) = - 1$
(resp. $( a / p ) = ( \mathrm{LM}_{\mathbb{F}_p} (a) / p ) = 1$)
generates a sequence of long period.
Note that $17 \equiv 1 \mathrm{~mod~} 4$ and $23 \equiv 3 \mathrm{~mod~} 4$.

We begin to show properties of logistic generator sequences.

\begin{lemma}
Let $a \in \mathbb{F}_p$.
\begin{enumerate}
\item
$\# \{ \mathrm{LM}_{\mathbb{F}_p}^{- 1} (a) \} = 2$ if and only if
$( a / p ) = ( \mathrm{LM}_{\mathbb{F}_p} (a) / p )$.
\item
$\# \{ \mathrm{LM}_{\mathbb{F}_p}^{- 1} (a) \} = 1$ if and only if $a = - 1$.
\item
Assume that $a \neq - 1$. 
Then $\# \{ \mathrm{LM}_{\mathbb{F}_p}^{-1} (a) \} = 0$ if and only if
$( a / p ) \neq ( \mathrm{LM}_{\mathbb{F}_p} (a) / p )$.
\end{enumerate}
\label{lemma:the number of inverse image}
\end{lemma}

\begin{proof}
For $a \in \mathbb{F}_p$, the discriminant of the polynomial $4 X (X + 1) - a$ is $D := 16 (a + 1)$.
Hence the statement follows from $( \mathrm{LM}_{\mathbb{F}_p} (a) / p ) = ( a / p ) ( D / p )$.
\end{proof}

By Lemma \ref{lemma:the number of inverse image},
every element in $D_0 \cap ( D_0 - 1 )$ and $D_1 \cap ( D_0 - 1 )$
has two inverse images of $\mathrm{LM}_{\mathbb{F}_p}$.
That is to say, these elements are candidates for being in cycles.

\begin{lemma}
Let $a \in D_0 \cap ( D_0 - 1 )$ or $a \in D_1 \cap ( D_0 - 1 )$.
Put $\{ c_1, c_2 \} = \mathrm{LM}_{\mathbb{F}_p}^{- 1} (a)$.
\begin{enumerate}
\item If $p \equiv 3 \mathrm{~mod~} 4$ and $a \in D_0 \cap ( D_0 - 1 )$, then $( c_1 / p ) \neq ( c_2 / p )$.
\item If $p \equiv 1 \mathrm{~mod~} 4$ and $a \in D_1 \cap ( D_0 - 1 )$, then $( c_1 / p ) \neq ( c_2 / p )$.
\end{enumerate}
\label{lemma:IV for long period}
\end{lemma}

\begin{proof}
Since $4 X ( X + 1 ) - a = 4 ( X - c_1 ) ( X - c_2 )$, $- a = 4 c_1 c_2$.
Hence we have
\begin{equation*}
\left( \frac{c_1}{p} \right) \left( \frac{c_2}{p} \right) = \left( \frac{- 1}{p} \right) \left( \frac{a}{p} \right).
\end{equation*}
The statement follows from the first supplement to quadratic reciprocity.
\end{proof}

By Lemma \ref{lemma:IV for long period},
if $p \equiv 3 \mathrm{~mod~} 4$ (resp. $p \equiv 1 \mathrm{~mod~} 4$),
then every element in $D_0 \cap ( D_0 - 1 )$ (resp. $D_1 \cap ( D_0 - 1 )$) is in a cycle.
Therefore one can expect that an element which satisfies the conditions
in Lemma \ref{lemma:IV for long period} generates a long period sequence.

In order to analyze $D_0 \cap ( D_0 - 1 )$ and $D_1 \cap ( D_0 - 1 )$, we introduce a structure of a hyperbola.
Let $C$ be the hyperbola over $\mathbb{F}_p$ defined by the equation $x^2 - y^2 = 1$.
For an extension field $K / \mathbb{F}_p$, $C(K)$ denotes the set of $K$-rational points on $C$.
Then we have a bijective map $\psi_K : K - \{ 0 \} \rightarrow C(K)$ defined as
$\psi_K (t) = ( 2^{-1} ( t + t^{- 1} ), 2^{-1} ( t - t^{- 1} ) )$
for any $t \in K - \{ 0 \}$ (See Silverman \cite[I.1.3.1]{Silverman}).

Assume that $p \equiv 3 \mathrm{~mod~} 4$.
Put $\mathrm{Param}_3 = \mathbb{F}_p - \{ 0, \pm 1 \}$.
We define a map $\Phi_3 : \mathrm{Param}_3 \rightarrow D_0 \cap ( D_0 - 1 )$ as
$\Phi_3 (t) = {\{ 2^{- 1} ( t - t^{- 1} ) \}}^2, \, t \in \mathrm{Param}_3$.

Assume that $p \equiv 1 \mathrm{~mod~} 4$.
Put
\begin{equation*}
\mathrm{Param}_1 = \{ t \in \mathbb{F}_{p^2} \mid \mathrm{N}_{\mathbb{F}_{p^2} / \mathbb{F}_p } (t) = 1 \} - \{ \pm 1 \}.
\end{equation*}
We define a map $\Phi_1 : \mathrm{Param}_1 \rightarrow D_1 \cap ( D_0 - 1 )$ as
$\Phi_1 (t) = {\{ 2^{- 1} ( t - t^{- 1} ) \}}^2, \, t \in \mathrm{Param}_1$.

\begin{proposition}
Let $p > 3$ be a prime number.
If $p \equiv 3 \mathrm{~mod~} 4$ (resp. $p \equiv 1 \mathrm{~mod~} 4$),
then $\Phi_3$ (resp. $\Phi_1$) is four-to-one map.
\end{proposition}

\begin{proof}
First we consider the case of $p \equiv 3 \mathrm{~mod~} 4$.
Let $a \in D_0 \cap ( D_0 - 1 )$.
Then there are $b, c \in \mathbb{F}_p$ such that $a = b^2$ and $a + 1 = c^2$. Hence $(c, b) \in C(\mathbb{F}_p)$.
Since $a \neq 0$, $b \neq 0$.
Hence we have four-to-one map $\Phi_3$.

Next we consider the case of $p \equiv 1 \mathrm{~mod~} 4$.
Let $a \in D_1 \cap ( D_0 - 1 )$.
Then there are $\beta \in \mathbb{F}_{p^2} - \mathbb{F}_p$ and $c \in \mathbb{F}_p$
such that $a = {\beta}^2$ and $a + 1 = c^2$.
Hence $(c, \beta) \in
\{ (u, v) \in C(\mathbb{F}_{p^2}) \mid u \in \mathbb{F}_p, \, v \in \mathbb{F}_{p^2} - \mathbb{F}_p \}$.
Now, we have
\begin{eqnarray*}
&& \{ t \in \mathbb{F}_{p^2} \mid t + t^{- 1} \in \mathbb{F}_p, \, t - t^{- 1} \in \mathbb{F}_{p^2} - \mathbb{F}_p \} \\
&=& \{ t \in \mathbb{F}_{p^2} - \mathbb{F}_p \mid N_{\mathbb{F}_{p^2} / \mathbb{F}_p} (t) = 1 \}.
\end{eqnarray*}
Note that $t \in \mathbb{F}_p$ and $N_{\mathbb{F}_{p^2} / \mathbb{F}_p} (t) = 1$ if and only if $t = \pm 1$.
Hence we have four-to-one map $\Phi_1$.
\end{proof}

\begin{remark}
For $j \in \{ 1, 3 \}$, we have
\begin{equation}
\Phi_j (t) = \Phi_j (- t) = \Phi_j (t^{- 1}) = \Phi_j (- t^{- 1}), \quad t \in \mathrm{Param}_j.
\label{equation:four to one map}
\end{equation}
\label{remark:four to one map}
\end{remark}

\begin{remark}
Let $\mathbb{G}_m$ be the multiplicative group scheme and
$T_2 = \mathrm{Ker} [ N_{\mathbb{F}_{p^2} / \mathbb{F}_p} : 
\mathrm{Res}_{\mathbb{F}_{p^2} / \mathbb{F}_p} \, \mathbb{G}_m \rightarrow \mathbb{G}_m ]$ the norm one torus
(For details, see \cite{Waterhouse}, \cite{Voskresenskii} and \cite{Rubin-Silverberg}).
Then we have $\mathrm{Param}_3 = \mathbb{G}_m ( \mathbb{F}_p ) - \{ \pm 1 \}$ and
$\mathrm{Param}_1 = T_2 ( \mathbb{F}_p ) - \{ \pm 1 \}$.
Thus, they have a common structure of a set of $\mathbb{F}_p$-rational points on an algebraic torus of dimension one
except elements of order $1$ and $2$.
\end{remark}

\begin{corollary}
Let $p > 3$ be a prime number.
If $p \equiv 3 \mathrm{~mod~} 4$ (resp. $p \equiv 1 \mathrm{~mod~} 4$),
then $\# \{ D_0 \cap ( D_0 - 1 ) \} = (p - 3) / 4$ (resp. $\# \{ D_1 \cap ( D_0 - 1 ) \} = (p - 1) / 4$).
\label{corollary:the number of sets of initial values}
\end{corollary}

Let $S = ( s_n )_{n \geq 0}$ be a sequence defined as the recurrence relation (\ref{equation:logistic generator})
with $s_0 \in D_0 \cap ( D_0 - 1 )$ (resp. $s_0 \in D_1 \cap ( D_0 - 1 )$)
if $p \equiv 3 \mathrm{~mod~} 4$ (resp. $p \equiv 1 \mathrm{~mod~} 4$).
By Corollary \ref{corollary:the number of sets of initial values},
the period of $S$ is upper bounded by $( p - 3 ) / 4$ (resp. $( p - 1 ) / 4$)
if $p \equiv 3 \mathrm{~mod~} 4$ (resp. $p \equiv 1 \mathrm{~mod~} 4$).

\begin{example}
Let $p = 23$.
Then we have $D_0 \cap ( D_0 - 1 ) = \{ 1, 2, 3, 8, 12 \}$ (See Fig. \ref{figure:state diagram of p = 23}).
Table \ref{table:Param and IVset p = 23} shows the corresponence between $\mathrm{Param}_3$ and $D_0 \cap ( D_0 - 1 )$.

\begin{table}[tb]
\caption{The corresponence between $\mathrm{Param}_3$ and $D_0 \cap ( D_0 - 1 )$ in the case of $p = 23$.}
\label{table:Param and IVset p = 23}
\begin{center}
\begin{tabular}{cc}
\hline\noalign{\smallskip}
$t \in \mathrm{Param}_3$ & $\Phi_3 (t) \in D_0 \cap ( D_0 - 1 )$ \\
\noalign{\smallskip}\hline\noalign{\smallskip}
$4, 6, 17, 19$ & $1$ \\
$2, 11, 12, 21$ & $2$ \\
$5, 9, 14, 18$ & $3$ \\
$7, 10, 13, 16$ & $8$ \\
$3, 8, 15, 20$ & $12$ \\
\noalign{\smallskip}\hline
\end{tabular}
\end{center}
\end{table}

\end{example}

\begin{example}
Let $p = 17$.
Then we have $D_1 \cap ( D_0 - 1 ) = \{ 3, 7, 12, 14 \}$ (See Fig. \ref{figure:state diagram of p = 17}).
Let $\alpha$ be a root of $X^2 - 3 \in \mathbb{F}_p [X]$.
Then $\mathbb{F}_p ( \alpha )$ is a quadratic extension of $\mathbb{F}_p$.
Table \ref{table:Param and IVset p = 17} shows the corresponence between $\mathrm{Param}_1$ and $D_1 \cap ( D_0 - 1 )$.

\begin{table}[tb]
\caption{The corresponence between $\mathrm{Param}_1$ and $D_1 \cap ( D_0 - 1 )$ in the case of $p = 17$.}
\label{table:Param and IVset p = 17}
\begin{center}
\begin{tabular}{cc}
\hline\noalign{\smallskip}
$t \in \mathrm{Param}_1$ & $\Phi_1 (t) \in D_1 \cap ( D_0 - 1 )$ \\
\noalign{\smallskip}\hline\noalign{\smallskip}
$2 + \alpha, 15 + 16 \alpha, 2 + 16 \alpha, 15 + \alpha$ & $3$ \\
$5 + 5 \alpha, 12 + 12 \alpha, 5 + 12 \alpha, 12 + 5 \alpha$ & $7$ \\
$8 + 2 \alpha, 9 + 15 \alpha, 8 + 15 \alpha, 9 + 2 \alpha$ & $12$ \\
$7 + 4 \alpha, 10 + 13 \alpha, 7 + 13 \alpha, 10 + 4 \alpha$ & $14$ \\
\noalign{\smallskip}\hline
\end{tabular}
\end{center}
\end{table}

\end{example}

For understanding periods, we relate the logistic map on the sets of initial values and
the square map on the parameter spaces of the hyperbola.

\begin{lemma}
If $p \equiv 3 \mathrm{~mod~} 4$ (resp. $p \equiv 1 \mathrm{~mod~} 4$),
then $\mathrm{LM}_{\mathbb{F}_p} ( \Phi_3 (t) ) = \Phi_3 ( t^2 ), \, t \in \mathrm{Param}_3$
(resp. $\mathrm{LM}_{\mathbb{F}_p} ( \Phi_1 (t) ) = \Phi_1 ( t^2 ), \, t \in \mathrm{Param}_1$).
\label{lemma:square map and logistic map}
\end{lemma}

\begin{proof}
Assume that $p \equiv 3 \mathrm{~mod~} 4$.
For $t \in \mathrm{Param}_3$,
\begin{eqnarray*}
\mathrm{LM}_{\mathbb{F}_p} ( \Phi_3 (t) )
&=& 4 \times {\{ 2^{- 1} ( t - t^{- 1} ) \}}^2 \times {\{ 2^{- 1} ( t + t^{- 1} ) \}}^2 \\
&=& {\{ 2^{- 1} ( t^2 - t^{- 2} ) \}}^2 \\
&=& \Phi_3 ( t^2 ).
\end{eqnarray*}
For the case of $p \equiv 1 \mathrm{~mod~} 4$, one can show the statement similarly.
\end{proof}

Lemma \ref{lemma:square map and logistic map} follows that
periods of logistic generator sequences on the sets of initial values are induced by
those of sequences generated by the square map on the parameter spaces of the hyperbola.
The latter has been studied by Rogers \cite{Rogers} and Vasiga and Shallit \cite{Vasiga-Shallit}.

\subsection{Periods of sequences}

The period of a logistic generator sequence and the state diagram given by the logistic map
are obtained by Theorem \ref{theorem:Vasiga-Shallit Theorem 12} and Theorem \ref{theorem:Vasiga-Shallit Corollary 15}.
However, we revisit similar ones because we give conditions to be maximal on the sets of initial values below.

\begin{theorem}
Let $p > 3$ be a prime number.
Let $t \in \mathrm{Param}_3$ (resp. $t \in \mathrm{Param}_1$) and
$S = ( s_n )_{n \geq 0}$ a sequence defined as the recurrence relation (\ref{equation:logistic generator})
with $s_0 = \Phi_3 (t) \in D_0 \cap ( D_0 - 1 )$ (resp. $s_0 = \Phi_1 (t) \in D_1 \cap ( D_0 - 1 )$)
if $p \equiv 3 \mathrm{~mod~} 4$ (resp. $p \equiv 1 \mathrm{~mod~} 4$).
Let $\mathrm{ord}_{G} \, t = 2^e \cdot m$, where $G = \mathbb{F}_{p^2}^{\times}$ and $m$ is odd.
Then the length of the tail of $S$ is
\begin{equation}
t(S) = \left\{
\begin{array}{ll}
0 & \mbox{($e = 0$)} \\
e - 1 & \mbox{($e > 0$)}
\end{array}
\right., 
\label{equation:tail length}
\end{equation}
and the period of $S$ is $c(S) = \mathrm{ord}_m' \, 2$
(For the definition of $\mathrm{ord}_m' \, 2$,
see Subsect. \ref{subsection:The Dickson generator of degree $2$}).
\end{theorem}

\begin{proof}
Note that the sequence $( t^{2^n} )_{n \geq 0}$ has period $\mathrm{ord}_m \, 2$ and tail length $e$
(Vasiga and Shallit \cite[Theorem 1]{Vasiga-Shallit}).
By (\ref{equation:four to one map}), we have $c(S) = \mathrm{ord}_m' \, 2$.
On the other hand, $\mathrm{ord}_{G} \, t = m$ for some odd integer $m$ if and only if
$\mathrm{ord}_{G} \, (- t) = 2 m$ for some odd integer $m$.
Hence we have (\ref{equation:tail length}).
\end{proof}

\begin{theorem}
Let $p > 3$ be a prime number.
If $p \equiv 3 \mathrm{~mod~} 4$ (resp. $p \equiv 1 \mathrm{~mod~} 4$),
suppose that $p - 1 = 2 m$ (resp. $p + 1 = 2 m$), where $m$ is an odd integer.
If $p \equiv 3 \mathrm{~mod~} 4$ (resp. $p \equiv 1 \mathrm{~mod~} 4$),
for any divisor $d \neq 1$ of $m$ the state diagram given by $\mathrm{LM}_{\mathbb{F}_p}$
on $D_0 \cap ( D_0 - 1 )$ (resp. $D_1 \cap ( D_0 - 1 )$)
consists of $n_d := \varphi(d) / (2 \mathrm{ord}_d' \, 2)$ cycles of period $c_d := \mathrm{ord}_d' \, 2$,
where $\varphi$ is Euler's totient function.
\label{theorem:structure of period}
\end{theorem}

\begin{proof}
For any divisor $d$ of $m$,
the state diagram given by the square map consists of $\varphi (d) / \mathrm{ord}_d \, 2$ cycles
of period $\mathrm{ord}_d \, 2$ on $\mathrm{Param}_3$ (resp. $\mathrm{Param}_1$)
if $p \equiv 3 \mathrm{~mod~} 4$ (resp. $p \equiv 1 \mathrm{~mod~} 4$).
For details, see Rogers \cite[Theorem]{Rogers} or Vasiga and Shallit \cite[Corollary 3]{Vasiga-Shallit}.

Let $d \neq 1$ be a divisor of $m$.

Suppose that there is no integer $k \in \mathbb{Z}$ such that $2^k \equiv - 1 \mathrm{~mod~} d$.
Then a cycle of period $c_d$ on $D_0 \cap ( D_0 - 1 )$ (resp. $D_1 \cap ( D_0 - 1 )$) corresponds to
two cycles of period $\mathrm{ord}_d \, 2$ on $\mathrm{Param}_3$ (resp. $\mathrm{Param}_1$).
Hence the state diagram consists of $\varphi (d) / ( 2 \mathrm{ord}_d \, 2 )$ cycles of period $\mathrm{ord}_d \, 2$.

Suppose that there is an integer $k \in \mathbb{Z}$ such that $2^k \equiv - 1 \mathrm{~mod~} d$.
Then a cycle of period $c_d$ on $D_0 \cap ( D_0 - 1 )$ (resp. $D_1 \cap ( D_0 - 1 )$) corresponds to
a cycle of period $\mathrm{ord}_d \, 2$ on $\mathrm{Param}_3$ (resp. $\mathrm{Param}_1$).
Hence the state diagram consists of $\varphi (d) / \mathrm{ord}_d \, 2$ cycles of period $\mathrm{ord}_d \, 2 / 2$.
\end{proof}

\begin{example}
Let $p = 23$. Put $m = 11$, then $p = 2 m + 1$. 
Table \ref{table:state diagram of p = 23} shows the structure of the state diagram
given by $\mathrm{LM}_{\mathbb{F}_p}$ on $D_0 \cap ( D_0 - 1 ) = \{ 1, 2, 3, 8, 12 \}$.
In fact, the state diagram consists of the cycle $(1, 8, 12, 3, 2)$ (See Fig. \ref{figure:state diagram of p = 23}).

\begin{table}[tb]
\caption{The structure of the state diagram given by $\mathrm{LM}_{\mathbb{F}_{23}}$ on $D_0 \cap ( D_0 - 1 )$.}
\label{table:state diagram of p = 23}
\begin{center}
\begin{tabular}{ccccc}
\hline\noalign{\smallskip}
$d$ & $\mathrm{ord}_d \, 2$ & $ \varphi (d)$ & $n_d$ & $c_d$ \\
\noalign{\smallskip}\hline\noalign{\smallskip}
$11$ & $10$ & $10$ & $1$ & $5$ \\
\noalign{\smallskip}\hline
\end{tabular}
\end{center}
\end{table}
\end{example}

\begin{example}
Let $p = 17$. Put $m = 9$, then $p = 2 m - 1$. 
Table \ref{table:state diagram of p = 17} shows the structure of the state diagram
given by $\mathrm{LM}_{\mathbb{F}_p}$ on $D_1 \cap ( D_0 - 1 ) = \{ 3, 7, 12, 14 \}$.
In fact, the state diagram consists of the cycles $(12)$ and $(3, 14, 7)$
(See Fig. \ref{figure:state diagram of p = 17}).

\begin{table}[tb]
\caption{The structure of the state diagram given by $\mathrm{LM}_{\mathbb{F}_{17}}$ on $D_1 \cap ( D_0 - 1 )$.}
\label{table:state diagram of p = 17}
\begin{center}
\begin{tabular}{ccccc}
\hline\noalign{\smallskip}
$d$ & $\mathrm{ord}_d \, 2$ & $ \varphi (d)$ & $n_d$ & $c_d$ \\
\noalign{\smallskip}\hline\noalign{\smallskip}
$3$ & $2$ & $2$ & $1$ & $1$ \\
$9$ & $6$ & $6$ & $1$ & $3$ \\
\noalign{\smallskip}\hline
\end{tabular}
\end{center}
\end{table}
\end{example}

\subsection{Long period sequences}

In order to apply the logistic map over finite fields to a pseudorandom number generator,
a generating sequence is required to have a long period.
Now, we show the conditions for parameters to be maximal on the sets of initial values.

\begin{corollary}
Let $p > 3$ be a prime number.
If $p \equiv 3 \mathrm{~mod~} 4$ (resp. $p \equiv 1 \mathrm{~mod~} 4$),
let $S = ( s_n )_{n \geq 0}$ be a sequence defined as the recurrence relation (\ref{equation:logistic generator})
with $s_0 \in D_0 \cap ( D_0 - 1 )$ (resp. $s_0 \in D_1 \cap ( D_0 - 1 )$).
Then $S$ attains a maximal period if and only if there exists a prime number $p_1$ such that
\begin{equation}
\left\{
\begin{array}{ll}
p = 2 p_1 + 1 & \mbox{if $p \equiv 3 \mathrm{~mod~} 4$} \\
p = 2 p_1 - 1 & \mbox{if $p \equiv 1 \mathrm{~mod~} 4$}
\end{array}
\right.
\label{equation:m = p_1}
\end{equation}
and
\begin{equation}
\mathrm{ord}_{p_1} \, 2 = p_1 - 1
\mbox{~or~}
\left\{
\begin{array}{l}
\mathrm{ord}_{p_1} \, 2 = (p_1 - 1) / 2 \\
(p_1 - 1) / 2 \quad \mbox{is an odd integer.}
\end{array}
\right.
\label{equation:primitive and half primitive element}
\end{equation}
\label{corollary:maximal sequence}
\end{corollary}

\begin{proof}
Suppose that $S$ attains a maximal period.
Assume that $p \equiv 3 \mathrm{~mod~} 4$.
Since the state diagram given by $\mathrm{LM}_{\mathbb{F}_p}$ on $D_0 \cap ( D_0 - 1 )$ has one cycle,
$p_1 := (p - 1) / 2$ is a prime number and
(\ref{equation:primitive and half primitive element}) holds by Theorem \ref{theorem:structure of period}.
Similarly, if $p \equiv 1 \mathrm{~mod~} 4$,
then $p_1 := (p + 1) / 2$ is a prime number and (\ref{equation:primitive and half primitive element}) holds.

Conversely, suppose that there exists a prime number $p_1$ such that
(\ref{equation:m = p_1}) and (\ref{equation:primitive and half primitive element}) hold.
Then $S$ has period $(p_1 - 1) / 2$ by Theorem \ref{theorem:structure of period}.
\end{proof}

Assume that $p \equiv 3 \mathrm{~mod~} 4$.
$p$ is called a maximal prime on $D_0 \cap ( D_0 - 1 )$
if there exists a prime number $p_1$ such that
(\ref{equation:m = p_1}) and (\ref{equation:primitive and half primitive element}) hold.
Assume that $p \equiv 1 \mathrm{~mod~} 4$.
$p$ is called a maximal prime on $D_1 \cap ( D_0 - 1 )$
if there exists a prime number $p_1$ such that
(\ref{equation:m = p_1}) and (\ref{equation:primitive and half primitive element}) hold.

\begin{remark}
For a prime number $p$, $p$ is a safe prime if there is a prime number $p_1$ such that $p = 2 p_1 + 1$.
For a safe prime $p = 2 p_1 + 1$, $p$ is a $2$-safe prime if $p_1$ is also a safe prime.
If $p$ is a $2$-safe prime, then $p$ is a maximal prime on $D_0 \cap ( D_0 - 1 )$.
Peinado, Montoya, Mu\~{n}oz and Yuste \cite[Remark 2]{Peinado-Montoya-Munoz-Yuste} and
Miyazaki, Araki, Uehara and Nogami \cite{Miyazaki-Araki-Uehara-Nogami:SCIS2014}
considered in this situation.
We show small $2$-safe primes as follows:
\begin{equation*}
11, 23, 47, 167, 359, 719, 1439, 2039, 2879, 4079.
\end{equation*}
\label{remark:2-safe prime}
\end{remark}

\begin{remark}
Let $p > 3$ be a prime number such that $p \equiv 1 \mathrm{~mod~} 4$.
Assume that there is a prime number $p_1$ such that $p = 2 p_1 - 1$ and $p_1$ is a safe prime
($p$ is analogous to a $2$-safe prime).
Then $p$ is a maximal prime on $D_1 \cap ( D_0 - 1 )$.
However, a prime number that satisfies these conditions is only $13$.
In fact, we may asuume that $p_1 > 3$.
Since $p \not\equiv 0 \mathrm{~mod~} 3$, $p_1 \equiv 1 \mathrm{~mod~} 3$.
Put $p_1 = 2 p_2 + 1$ for some prime number $p_2$.
Since $p_1 \equiv 1 \mathrm{~mod~} 3$, $p_2 \equiv 0 \mathrm{~mod~} 3$, that is $p_2 = 3$.
Hence $p = 4 p_2 + 1 = 13$.
\label{remark:p = 13}
\end{remark}

In Corollary \ref{corollary:maximal sequence},
we state the necessary and sufficient condition for initial values and parameters to be maximal.
Now, we estimate density for maximal sequences by numerical experiments.

We show the graphs of
the percentage of $N$-bit maximal primes on $D_0 \cap ( D_0 - 1 )$ in $N$-bit primes congruent to $3$ modulo $4$ and of
the percentage of $N$-bit maximal primes on $D_1 \cap ( D_0 - 1 )$ in $N$-bit primes congruent to $1$ modulo $4$
in Fig. \ref{figure:the percentage of maximal primes} for $3 \leq N \leq 32$.
From Fig. \ref{figure:the percentage of maximal primes},
it is observed that the both percentages are slowly decreasing as $N$ grows.
Moreover,
the percentage of $N$-bit maximal primes on $D_0 \cap ( D_0 - 1 )$ in $N$-bit primes congruent to $3$ modulo $4$
is larger than that of $N$-bit maximal primes on $D_1 \cap ( D_0 - 1 )$ in $N$-bit primes congruent to $1$ modulo $4$ for $N \geq 12$.
The cause of this phenomenon is not clear,
however Remark \ref{remark:2-safe prime} and Remark \ref{remark:p = 13} provide pieces of evidence.

\begin{figure}[tb]
\begin{center}
\includegraphics[width=3.0in]{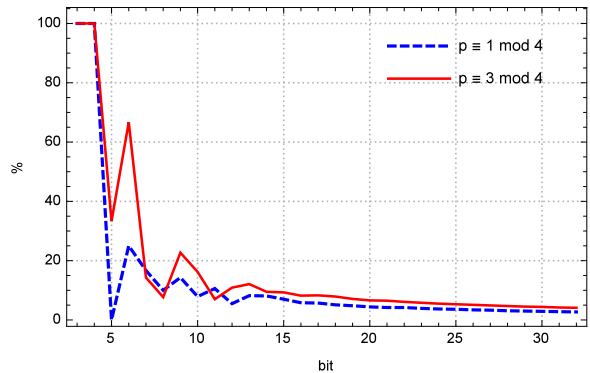}
\end{center}
\caption{The graphs of the percentage of maximal primes.}
\label{figure:the percentage of maximal primes}
\end{figure}

This observation shows that maximal primes are rare.
Therefore we estimate the number of cycles and their periods on the sets of initial values
from Theorem \ref{theorem:structure of period}.
We show the graphs of the average of the number of cycles in Fig. \ref{figure:the average of the number of cycles} and
the graphs of the average of periods in Fig. \ref{figure:the average of periods} for $17 \leq N \leq 32$.
From Fig. \ref{figure:the average of the number of cycles} and Fig. \ref{figure:the average of periods},
it is observed that the both averages of the number of cycles and of periods are increasing rapidly as $N$ grows.
As in Fig. \ref{figure:the percentage of maximal primes},
the average of periods in the case of $N$-bit primes congruent to $3$ modulo $4$ is larger than
that in the case of $N$-bit primes congruent to $1$ modulo $4$.
On the other hand,
the average of the number of cycles in the case of $N$-bit primes congruent to $3$ modulo $4$ is almost the same as
that in the case of $N$-bit primes congruent to $1$ modulo $4$.

\begin{figure}[tb]
\begin{center}
\includegraphics[width=3.0in]{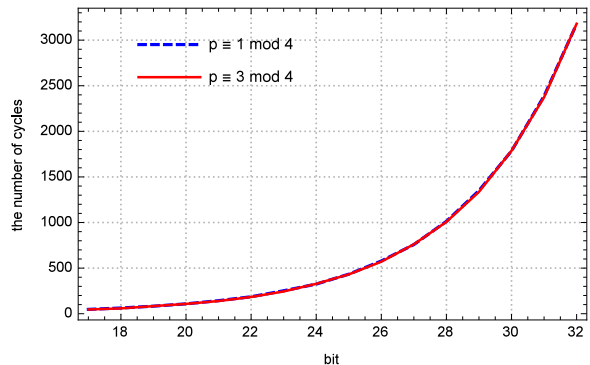}
\end{center}
\caption{The graphs of the average of the number of cycles.}
\label{figure:the average of the number of cycles}
\end{figure}

\begin{figure}[tb]
\begin{center}
\includegraphics[width=3.0in]{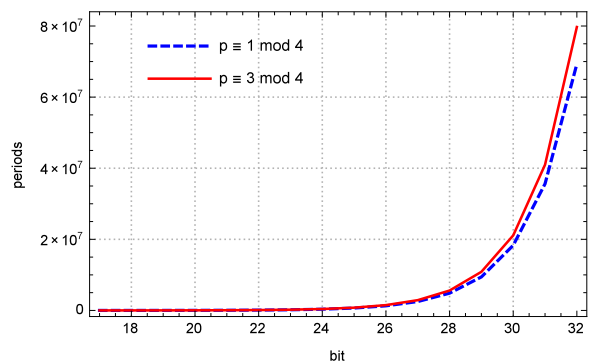}
\end{center}
\caption{The graphs of the average of periods.}
\label{figure:the average of periods}
\end{figure}

\section{Remarks on the linear complexity profile of the logistic generator}
\label{section:Remarks on the linear complexity profile of the logistic generator}

In this section, we remark on the linear complexity profile of the logistic generator.
In particular, we show two lower bounds on the linear complexity profile and compare the two lower bounds.

\subsection{Definitions}

We recall the definition of linear complexity profile.
For details, see \cite[10.4]{Mullen-Panario}.

Let $q$ be a prime power.
Let $S = ( s_n )_{n \geq 0}$ be a sequence over $\mathbb{F}_q$.
The linear complexity $L(S)$ of $S$ is the length $L$ of the shortest linear recurrence relation
\begin{equation*}
s_{n + L} = a_{L - 1} s_{n + L - 1} + \cdots + a_{0} s_{n}, \quad n \geq 0
\end{equation*}
for some $a_0, \dots ,a_{L - 1} \in \mathbb{F}_q$.
For $N \geq 1 \in \mathbb{Z}$,
the $N$-th linear complexity $L(S, N)$ of $S$ is the length $L$ of the shortest linear recurrence relation
\begin{equation*}
s_{n + L} = a_{L - 1} s_{n + L - 1} + \cdots + a_{0} s_{n}, \quad 0 \leq n \leq N - L -1
\end{equation*}
for some $a_0, \dots ,a_{L - 1} \in \mathbb{F}_q$.
The linear complexity profile of $S$ is the sequence $( L(S, N) )_{N \geq 1}$ of integers.

Assume that $S$ is periodic of length $T$.
Put $s^T (X) = s_0 + s_1 X + \cdots + s_{T - 1} X^{T - 1} \in \mathbb{F}_q [X]$.
Then $L(S) = L(S, 2 T)\leq T$ and $L(S) = T - \mathrm{deg} \, \mathrm{GCD} \ (X^T - 1, s^T (X))$
(See \cite[Theorem 10.4.27]{Mullen-Panario}).

\subsection{Lower bounds on the linear complexity profile}

Aly and Winterhof obtained a lower bound on the linear complexity profile of a Dickson generator as follows:

\begin{theorem}[\cite{Aly-Winterhof} Theorem 1]
Let $p > 2$ be a prime number.
Let $S = ( s_n )_{n \geq 0}$ be a sequence defined as the recurrence relation (\ref{equation:Dickson generator}).
Assume that $S$ is periodic of length $T$.
Then the lower bound
\begin{equation}
L(S, N) \geq \frac{ \min \{ N^2, 4 T^2 \} }{16 (p + 1)} - (p + 1)^{1 / 2}, \quad N \geq 1
\label{equation:LCP of a Dickson generator}
\end{equation}
holds.
\end{theorem}

Let $p > 3$ be a prime number.
Let $S = ( s_n )_{n \geq 0}$ be a sequence
defined as the recurrence relation (\ref{equation:logistic generator}) over $\mathbb{F}_p$.
Since $S$ is obtained from a Dickson generator sequence of degree $2$ with a linear transformation,
the linear complexity profile of $S$ satisfies the lower bound (\ref{equation:LCP of a Dickson generator}).
However, we obtain a little refined lower bound on the linear complexity profile of $S$.

\begin{theorem}
Let $m > 2$ be an odd integer such that $p := 2 m + 1$ (resp. $p := 2 m - 1$) is a prime number and
$S = ( s_n )_{n \geq 0}$ a sequence defined as the recurrence relation (\ref{equation:logistic generator})
with $s_0 \in D_0 \cap ( D_0 - 1 )$ (resp. $s_0 \in D_1 \cap ( D_0 - 1 )$) over $\mathbb{F}_p$.
Then the lower bound
\begin{equation}
L(S, N) \geq \frac{ \min \{ N^2, 4 T^2 \} }{16 m} - m^{1 / 2}, \quad N \geq 1
\label{equation:LCP 1}
\end{equation}
holds, where $T$ is the period of $S$.
\label{theorem:LCP 1}
\end{theorem}

\begin{proof}
Exactly like the proof of \cite[Theorem 1]{Aly-Winterhof}.
See \ref{section:Proof of LCP 1}.
\end{proof}

Since the lower bound (\ref{equation:LCP 1}) is larger than
the lower bound (\ref{equation:LCP of a Dickson generator}),
we can ensure a higher security level for a Dickson generator sequence of degree $2$.

In the case of arbitrary control parameter,
we obtain another lower bound on the linear complexity profile
by the same approach as the proof in Griffin and Shparlinski \cite[Theorem 7]{Griffin-Shparlinski}.

\begin{theorem}
Let $p > 2$ be a prime number and $\mu_p \in \mathbb{F}_p - \{ 0 \}$.
Let $S = ( s_n )_{n \geq 0}$ be a sequence defined as the recurrence relation
$s_{n + 1} = \mathrm{LM}_{\mathbb{F}_p [\mu_p]} ( s_n ), \, n \geq 0$.
Assume that $S$ is periodic.
Then the lower bound
\begin{equation}
L(S, N) \geq \min \left\{ ( 2 N )^{1 / 2} - 3, \, L(S) \right\}, \quad N \geq 1
\label{equation:LCP 2}
\end{equation}
holds.
\label{theorem:LCP 2}
\end{theorem}

\begin{proof}
Exactly like the proof of \cite[Theorem 7]{Griffin-Shparlinski}.
See \ref{section:Proof of LCP 2}.
\end{proof}

Now, we compare the lower bound (\ref{equation:LCP 1}) with the lower bound (\ref{equation:LCP 2}).
Since the order of the lower bound (\ref{equation:LCP 1}) in terms of $N$ is larger than
that of the lower bound (\ref{equation:LCP 2}),
the graph of the lower bound (\ref{equation:LCP 1}) is upper than
that of the lower bound (\ref{equation:LCP 2}) for large $N$.
In the case of maximal primes, if the size of $p$ is large enough,
the lower bound (\ref{equation:LCP 1}) is close to the period asymptotically as $N$ grows.
On the other hand, the lower bound (\ref{equation:LCP 1}) does not make sense for small $N$
because the estimated values are negative integers by the term $- m^{1 / 2}$ of (\ref{equation:LCP 1}).
Thus the two lower bounds are complementary to each other.
Figure \ref{figure:LCP 1 vs LCP 2} describes the graphs of the two lower bounds in the case of $p = 6599$.

\begin{figure}[tb]
\begin{center}
\includegraphics[width=3.0in]{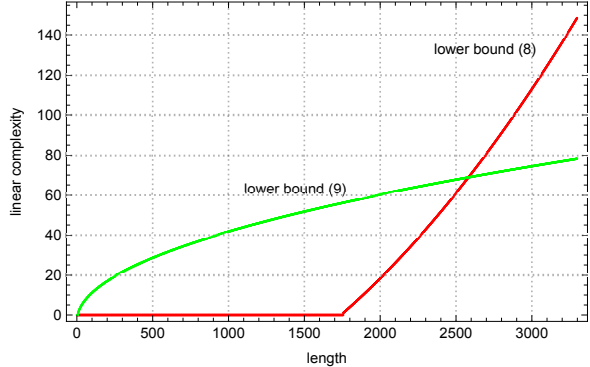}
\end{center}
\caption{The graphs of the lower bound (\ref{equation:LCP 1}) and
the lower bound (\ref{equation:LCP 2}) in the case of $p = 6599$.
If an estimated value is a negative integer, then we plot zero instead of the value.}
\label{figure:LCP 1 vs LCP 2}
\end{figure}

\section{A possibility of generalization}
\label{section:A possibility of generalization}

In this section, we consider a possibility of generalization.

Let $p > 2$ be a prime number and $\mu_p \in \mathbb{F}_p - \{ 0 \}$. Let $a \in \mathbb{F}_p$.
Then the discriminant of the polynomial $\mu_p X (X + 1) - a$ is $D_{\mu_p} := \mu_p^2 + 4 \mu_p a$.
If $\mu_p = 4$, then $( \mathrm{LM}_{\mathbb{F}_p} (a) / p )$ is described by $( D_4 / p )$
as in the proof of Lemma \ref{lemma:the number of inverse image}.
On the other hand, if $\mu_p \neq 4$,
then $( \mathrm{LM}_{\mathbb{F}_p} (a) / p )$ is not described by $( D_{\mu_p} / p )$.
Therefore, it is difficult to apply our methods to the case of $\mu_p \neq 4$.
However, note that a logistic generator with $\mu_p = 4$
is linearly transformed to a logistic generator with $\mu_p = p - 2$
by \cite[Theorem 1]{Miyazaki-Araki-Uehara-Nogami:ISITA2014}.
We do not know whether one can apply our methods to a general quadratic congruential generator.
Note that Peinado, Montoya, Mu\~{n}oz and Yuste \cite{Peinado-Montoya-Munoz-Yuste} showed
theoretical results about upper bounds on a period of the quadratic congruential generator sequence defined by
the quadratic polynomial $X^2 - c \in \mathbb{F}_p [X]$.

From another point of view,
one can consider a generalization of our methods to the Dickson generator of degree $n \geq 2$.
Unfortunately, how to apply our methods to the case of degree $n \geq 3$ is not clear.

\section{Conclusion}
\label{section:Conclusion}

In this paper, we investigate periods of logistic generator sequences with control parameter four.
In particular, we show the conditions for initial values and parameters to be maximal, and
estimate the percentage of maximal primes, the number of cycles and their periods on the sets of initial values.
Therefore, we can ensure that generating sequences of Dickson generator of degree $2$ have long period.
Moreover, we obtain a little refined lower bound on the linear complexity profile of these sequences.
As a consequence, the Dickson generator of degree $2$ has some good properties for cryptographic applications.

\section*{Acknowledgments}

The authors would like to thank Satoshi Uehara, Shunsuke Araki and Takeru Miyazaki for useful discussion.
In particular, the authors would like to thank Satoshi Uehara for his valuable comments.
This research was supported by JSPS KAKENHI Grant-in-Aid for Scientific Research (A) Number 16H01723.

\bibliographystyle{abbrv}
\bibliography{LogisticSeq}

\appendix

\section{Proof of Theorem \ref{theorem:LCP 1}}
\label{section:Proof of LCP 1}

We may assume that $N \leq 2 T$. Put
\begin{equation*}
t = m, \quad H = h = \left\lfloor \frac{N + 1}{2} \right\rfloor .
\end{equation*}
Then we have
\begin{equation*}
t = m \geq \mathrm{ord}_d \, 2 \geq T \geq \left\lfloor \frac{N + 1}{2} \right\rfloor = H = h
\end{equation*}
for some divisor $d \neq 1$ of $m$ (See Theorem \ref{theorem:structure of period}).
By \cite[Lemma 1]{Griffin-Shparlinski},
we see that the number $U$ of solutions of the congruence
\begin{equation*}
2^x \equiv y \mathrm{~mod~} t, \quad 0 \leq x \leq H - 1, \quad 0 \leq y \leq h - 1
\end{equation*}
satisfies
\begin{equation}
U \geq \frac{H h}{4 m} - m^{1 / 2} \geq \frac{N^2}{16 m} - m^{1 / 2}.
\label{equation:lower bound of the number of solutions}
\end{equation}
Let $(j_1, k_1), \dots ,(j_U, k_U)$ be the corresponding solutions.

Assume that $L(S, N) \leq U - 1$.
Let $W = ( w_n )_{n \geq 0}$ be a sequence of linear complexity $L(W) = L(S, N)$
with $w_n = s_n$ for $0 \leq n \leq N - 1$.
By \cite[Lemma 2]{Griffin-Shparlinski},
there exist $c_1, \dots ,c_U \in \mathbb{F}_p$, not all equal to zero, such that
\begin{equation*}
\sum_{i = 1}^{U} c_i w_{n + j_i} = 0, \quad n \geq 0.
\end{equation*}

For any $i \geq 0$, we have
\begin{equation*}
s_i = \{ 2^{- 1} ( t^{2^i} - t^{- 2^i} ) \}^2 = 4^{- 1} \left( t^{2^{i + 1}} + t^{- 2^{i + 1}} - 2 \right)
\end{equation*}
for some $t \in \mathrm{Param}_3$ (resp. $t \in \mathrm{Param}_1$)
if $p = 2 m + 1$ and $s_0 \in D_0 \cap ( D_0 - 1 )$ (resp. $p = 2 m - 1$ and $s_0 \in D_1 \cap ( D_0 - 1 )$).
Put $b_i = t^{2^{i + 1}} + t^{- 2^{i + 1}} \in \mathbb{F}_p$ for $i \geq 0$.
Since the discriminant $D_i$ of the polynomial $F_{b_i} (X) = X^2 - b_i X + 1 \in \mathbb{F}_p [X]$ is
$D_i = b_i^2 - 4 = 16 s_i (s_i + 1)$,
we have $( D_i / p ) = ( 16 s_i (s_i + 1) / p ) = ( \mathrm{LM}_{\mathbb{F}_p} (s_i) / p)$.

If $p = 2 m + 1$ and $s_0 \in D_0 \cap ( D_0 - 1 )$, then $( D_i / p ) = ( \mathrm{LM}_{\mathbb{F}_p} (s_i) / p) = 1$.
Hence $F_{b_i} (X)$ has two roots $t^{2^{i + 1}}, t^{- 2^{i + 1}} \in \mathbb{F}_p$ and
$( t^{2^{i + 1}} )^m = ( t^{2^i} )^{p - 1} = 1$.
Thus $D_e ( b_i, 1 ) = D_f ( b_i, 1 )$ if $e \equiv f \mathrm{~mod~} m$.
If $p = 2 m - 1$ and $s_0 \in D_1 \cap ( D_0 - 1 )$, then $( D_i / p ) = ( \mathrm{LM}_{\mathbb{F}_p} (s_i) / p) = - 1$.
Hence $F_{b_i} (X)$ has two roots $t^{2^{i + 1}}, t^{- 2^{i + 1}} \in \mathbb{F}_{p^2} - \mathbb{F}_p$ and
$( t^{2^{i + 1}} )^m = ( t^{2^i} )^{p + 1} = 1$.
Thus $D_e ( b_i, 1 ) = D_f ( b_i, 1 )$ if $e \equiv f \mathrm{~mod~} m$
(See \cite[Lemma 6]{Gomez-Perez-Gutierrez-Shparlinski}).

We have
\begin{eqnarray*}
w_{n + j_i} = s_{n + j_i}
&=& 4^{- 1} \left( t^{2^{n + 1 + j_i}} + t^{- 2^{n + 1 + j_i}} - 2 \right) \\
&=& 4^{- 1} \left( D_{ 2^{j_i} } ( t^{2^{n + 1}} + t^{- 2^{n + 1}}, 1 ) - 2 \right) \\
&=& 4^{- 1} \left( D_{ k_i } ( b_{n}, 1 ) - 2 \right)
\end{eqnarray*}
for $0 \leq n \leq N - 1 - j_i, \, 1 \leq i \leq U$.
Put
\begin{equation*}
f(X) = \sum_{i = 1}^{U} 4^{- 1} c_i \left( D_{ k_i } ( X, 1 ) - 2 \right) \in \mathbb{F}_p [X].
\end{equation*}
Then $f(X)$ is non-zero polynomial of degree
\begin{equation}
\mathrm{deg} \, f(X) \leq \max_{1 \leq i \leq U} k_i \leq h - 1.
\label{equation:degree of f}
\end{equation}
On the other hand, $f(X)$ has at least
\begin{equation*}
\min \left\{ T, N - \max_{1 \leq i \leq U} j_i \right\} \geq H = h
\end{equation*}
distinct zeros $b_{n}, \, 0 \leq n \leq \min \{ T, N - \max_{1 \leq i \leq U} j_i \} - 1$.
It contradicts (\ref{equation:degree of f}).
Hence we have $L(S, N) \geq U$ and
(\ref{equation:LCP 1}) follows from (\ref{equation:lower bound of the number of solutions}).

\section{Proof of Theorem \ref{theorem:LCP 2}}
\label{section:Proof of LCP 2}

Let $W = ( w_n )_{n \geq 0}$ be a sequence of the linear complexity $L(W) = L(S, N)$
with $w_n = s_n$ for $0 \leq n \leq N - 1$.
Let $U = ( u_n )_{n \geq 0}$ be the sequence defined as
$u_n = w_{n + 1} - \mathrm{LM}_{\mathbb{F}_p [\mu_p]} ( w_n ), \, n \geq 0$.
Then we have
\begin{eqnarray*}
L(U) &\leq& L( ( w_{n + 1} )_{n \geq 0} ) + \frac{L(W) ( L(W) + 1 )}{2} + L(W) \\
&=& 2 L(S, N) + \frac{L(S, N) ( L(S, N) + 1 )}{2}.
\end{eqnarray*}
Unless $U$ is the all-zeros sequence,
\begin{equation*}
2 N \leq 4 L(S, N) + L(S, N) ( L(S, N) + 1 ) \leq ( L(S, N) + 3 )^2.
\end{equation*}
Noting that $L(S, N) = L(S)$ if $U$ is the all-zeros sequence,
we have (\ref{equation:LCP 2}).

%

\end{document}